\documentclass[onecolumn,11pt,draftcls]{IEEEtran}

\usepackage{graphicx,amsfonts,amssymb,amsmath}
\usepackage{stackrel,bm}
\usepackage{cite}
\usepackage{mysymbols}

\begin{document}

\title{Connections between the Generalized Marcum $Q$-Function and a class of Hypergeometric Functions
\thanks{This work has been submitted to the IEEE for possible publication. Copyright may be transferred without notice, after which this version may no longer be accessible}}

\author{D. Morales-Jimenez,
        F. J. Lopez-Martinez,
        E. Martos-Naya,
        J. F. Paris,
        and A. Lozano
\thanks{D. Morales-Jimenez and A. Lozano are with the Department of Information and Communication Technologies, Universitat Pompeu Fabra, 08018 Barcelona (Spain). (e-mail:\{d.morales,angel.lozano\}@upf.edu)} 
\thanks{F. J. Lopez-Martinez is with the Wireless Systems Lab, Stanford University, Stanford, CA 94305-9515 (USA). (e-mail: fjlm@stanford.edu)}
\thanks{E. Martos-Naya and J. F. Paris are with Dept. Ingenieria de Comunicaciones, University of Malaga, 29071 Malaga (Spain). (e-mail: \{eduardo,paris\}@ic.uma.es)}}

\maketitle

\begin{abstract}
This paper presents a new connection between the generalized Marcum-$Q$ function and the confluent hypergeometric function of two variables, $\Phi_3$.
This result is then applied to the closed-form characterization of the bivariate Nakagami-$m$ distribution and of the distribution of the minimum eigenvalue of correlated non-central Wishart matrices, both important in communication theory. New expressions for the corresponding cumulative distributions are obtained and a number of communication-theoretic problems involving them are pointed out.
\end{abstract}


\section{Introduction} \label{sec:intro}

The number of special functions that make appearances in the communication theory arena continues to grow. Some of these are tabulated and well studied, with readily available approximations, bounds, and asymptotic expansions. Other special functions, however, are not so well studied and their numerical computation is not always as accurate and efficient as would be desirable. In such cases, the establishment of connections with other special functions can greatly reinforce the analytical toolbox available to researchers.


A wealth of works have studied the Marcum-$Q$ function (see, e.g., \cite{DiBlasio1991,Baricz2009,Corazza2002,Kapinas2009,YinSun2010}), deriving useful bounds and approximations and evidencing applications thereof. Routines for its accurate and efficient evaluation have been extensively developed and, in fact, the generalized Marcum-$Q$ function is included in most common mathematical software packages. 

  
Confluent hypergeometric functions also appear in a fair number of problems within communication theory and signal processing \cite{Alfano2010,Heath2005,Dharmawansa2011,Morales2012} \cite{Chatelain2007}. Moreover, recent distributional results \cite{Dharmawansa2011,Lopez2013} show that a number of other such problems (see, e.g.,  \cite{Narasimhan2003,Heath2005,Burel2002,
Simon05,Tan1997,Reig2002,Souza2008,Beaulieu2011,Abu-Dayya1994,Jakes1974,Zorzi1998,Wang1996}) can be solved in terms of confluent hypergeometric functions. 
Chief among these stands the $\Phi_3$ confluent hypergeometric function \cite{Gradsh00,Erdelyi1953}, which does not lend itself to easy and precise evaluation.


This paper shows that $\Phi_3$ can be expressed in terms of the generalized Marcum-$Q$ function. This connection is then applied to the closed-form characterization of two important distributions in communication-theory: the bivariate Nakagami-$m$ distribution, and the distribution of the minimum eigenvalue of correlated non-central Wishart matrices. New expressions for the corresponding cumulative distribution functions (CDFs) are obtained and some communication-theoretic problems involving them are pointed out. In particular, the connection unveiled herein, in combination with recent results in \cite{Lopez2013}, settles the standing conjecture \cite{Tan1997}\cite[p. 174]{Simon05} that the bivariate Nakagami-$m$ CDF can be expressed in terms of the generalized Marcum-$Q$ function.


\section{Preliminaries} \label{sec:preliminaries}

This section introduces the special functions under study and recalls a few properties of interest for the derivations that follow.

\begin{definition}
The generalized Marcum-$Q$ function is defined as \cite{Marcum1950, Nuttall1975, Simon05}
\begin{align} \label{eq:marcum}
Q_m \left( {a,b} \right) = \int_b^\infty  {\frac{{x^m }}{{a^{m - 1} }}\exp \left( { - \frac{{a^2  + x^2 }}{2}} \right)I_{m - 1} \left( {ax} \right)dx}
\end{align}
where $a>0$ and $b \geq 0$ are real parameters and $I_m(\cdot)$ is the $m$th order modified Bessel function of
the first kind.
The order index $m$ is an integer and typically $m \geq 0$, yet \refE{marcum} holds for negative orders too and a useful relationship between Marcum-$Q$ functions with positive and negative orders has been reported in \cite{ODriscoll2009}, namely
\begin{align} \label{eq:marcumNegOrders}
Q_m \left( {a,b} \right) = 1 - Q_{1 - m} \left( {b,a} \right) .
\end{align}
%
\end{definition}

\begin{definition}
The $\Phi_3$ confluent hypergeometric function of two variables is defined as \cite[Eq. 9.261.3]{Gradsh00} 
\begin{align} \label{eq:phi3}
\Phi _3 \left( {b,c;w,z} \right) = \sum\limits_{k = 0}^\infty  {\sum\limits_{m = 0}^\infty  {\frac{{\left( b \right)_k }}{{\left( c \right)_{k + m} }}\frac{{w^k z^m }}{{k!m!}}} }
\end{align}
where $b, c, w, z \in \real$, $c \neq 0,-1,-2,...$, and $(t)_r = \frac{\Gamma(t+r)}{\Gamma(t)}$ denotes the Pochhammer symbol with $\Gamma (\cdot)$ the Gamma function. The $\Phi_3$ function is one of the bivariate forms of the confluent hypergeometric function ${}_1F_1 (\cdot,\cdot;\cdot)$ \cite{Gradsh00}. Note that $\Phi_3$ does not exist for non-positive integer values of $c$ due to the singularities of the Gamma function. Next, we introduce a regularized version of this function, which is valid for any $c \in \real$.
\end{definition}

\begin{definition}
The regularized $\Phi_3$ function is defined as
\begin{align} \label{eq:regPhi3}
\tilde \Phi _3 \left( {b,c;w,z} \right) &= \frac{1}{{\Gamma \left( c \right)}}\Phi _3 \left( {b,c;w,z} \right) \nonumber \\
&= \sum\limits_{k = 0}^\infty  {\sum\limits_{m = 0}^\infty  {\frac{{\left( b \right)_k }}{{\Gamma \left( {c + k + m} \right)}}\frac{{w^k z^m }}{{k!m!}}} }
\end{align}
with $b, c, w, z \in \real$. For the special cases $w=0$ and/or $z=0$, $\tilde \Phi_3$ reduces to
\begin{align} \label{eq:Phi3_specialCases1}
\tilde \Phi _3 \left( {b,c;0,z} \right) &= z^{(1 - c)/2} I_{c - 1} \left( {2\sqrt z } \right) \\ \label{eq:Phi3_specialCases2}
\tilde \Phi _3 \left( {b,c;w,0} \right) &= {}_1\tilde F_1 \left( {b,c;w} \right) \\ \label{eq:Phi3_specialCases3}
\tilde \Phi _3 \left( {b,c;0,0} \right) &= \frac{1}{\Gamma (c)} 
\end{align}
where ${}_1 \tilde F_1 (b,c;w) = {}_1 F_1(b,c;w) / \Gamma(c) $ is the regularized confluent hypergeometric function \cite{Gradsh00}.
\end{definition}

The Laplace transforms of $\Phi_3$ and $\tilde \Phi _3$ are known and, due to their simple form, will be crucial in the ensuing derivations. Given the function $f(t) = {t^{c - 1} \tilde \Phi _3 \left( {b,c;xt,yt} \right)}$, its Laplace transform is given by \cite[Eq. 4.24.9]{Erdelyi1953}
\begin{align} \label{eq:LaplaceRegPhi3}
\Lcal \left\{ {f(t);s} \right\} = s^{ - c} \left( {1 - \frac{x}{s}} \right)^{ - b} e^{y/s} .
\end{align}

\section{Main Result} \label{sec:results}

The main result, presented in this section, rests on two new lemmas that provide, respectively, a new representation for the generalized Marcum-$Q$ function and a recursive relationship for $\tilde \Phi_3$.    

\begin{lemma} \label{lemma1}
The generalized Marcum-$Q$ function can be expressed in terms of $\tilde \Phi_3$ as
\begin{align} \label{eq:marcum_phi3}
Q_m \left( {a,b} \right) &= \left( {\frac{{a^2 }}{2}} \right)^{1 - m} \exp{ \left( - \frac{{a^2  + b^2 }}{2} \right)} \tilde \Phi _3 \left( {1,2 - m;\frac{{a^2 }}{2},\frac{{a^2 b^2 }}{4}} \right) , \qquad m \in \integer .
\end{align}
\end{lemma}

\begin{proof}
See Appendix \ref{apx:apx1}.
\end{proof}

\begin{lemma} \label{lemma3}
The regularized $\Phi_3$ function can be obtained recursively as
\begin{align} \label{eq:phi3_recursive}
\tilde \Phi _3 \left( {b,c;w,z} \right) = \left( \frac{z}{w} \right)^{b-1} \sum\limits_{i = 0}^{2(b - 1)} 
\frac{1}{z^i} \Acal _i (b,c;z) \tilde \Phi _3 \left( {1,c - i;w,z} \right)
\end{align}
for any $b \in \integer$ and $b > 0$, with $\Acal_i (b,c;z)$ being the polynomial on $z$ given by
\begin{align} \label{eq:phi3_coeffs}
\Acal _i (b,c;z)  = \frac{( - 1)^{b - 1} }{\left( b-1 \right)!} \sum\limits_{k = 0}^{\left\lfloor {i/2} \right\rfloor } {\frac{{( - 1)^k \left( {b - i + k} \right)_{i - k} \left( {c - i - 1 + k} \right)_{i - 2k} }}{{\left( {i - 2k} \right)!k!}} z^{k} } .
\end{align}
\end{lemma}

\begin{proof}
See Appendix \ref{apx:apx2}.
\end{proof}

Leveraging the foregoing lemmas, the main result in this paper can be put forth.

\begin{thm} \label{mainThm}
The $\tilde \Phi_3$ function is given in terms of the Marcum-$Q$ function as
\begin{align} \label{eq:phi3_marcum}
\tilde \Phi _3 \left( {b,c;w,z} \right) = \left( \frac{z}{w} \right)^{b-1} \sum\limits_{i = 0}^{2(b - 1)}  \frac{\Acal _i (b,c;z)}{ w^{c - i - 1} z^i } \exp{\left( w + \frac{z}{w} \right)}
Q_{2-c+i} \left( \sqrt{2w}, \sqrt{2\frac{z}{w}} \right) 
\end{align}
with $b,c \in \integer$, $b > 0$, $z \neq 0$, $w \neq 0$, and $\Acal_i (b,c;z)$ given by \refE{phi3_coeffs}.
\end{thm} 

\begin{proof}
The result follows directly from lemmas \ref{lemma1} and \ref{lemma3} after solving for $\tilde \Phi _3 \left( {1,c - i;w,z} \right)$ in \refE{marcum_phi3} and substituting in \refE{phi3_recursive}. 
\end{proof}

A few comments on the above theorem are in order:
\begin{itemize}
\item For the special cases $w=0$ and $z=0$, simpler connections are respectively given in   \refE{Phi3_specialCases1} and \refE{Phi3_specialCases2} in terms of $I_m(\cdot)$ and ${}_1 F_1(\cdot,\cdot;\cdot)$. If $w=z=0$, then \refE{Phi3_specialCases3} gives $\tilde \Phi_3$ in terms of the Gamma function. 

\item Albeit $\tilde \Phi _3 \left( {b,c;w,z} \right)$ is defined for any real value of its arguments, Thm. \ref{mainThm} is restricted to integer values of $b$ and $c$; these are precisely the cases of interest in communication theory. Moreover, the validity of \refE{phi3_marcum} can be straightforwardly extended to $c \in \real$ by applying \refE{marcum_phi3_2} in Appendix \ref{apx:apx1} in place of \refE{marcum_phi3}.
\item Negative values of $w,z$ (which do not correspond to known communication theory problems) imply complex arguments of the Marcum-$Q$ function. This is not an issue since the Marcum-$Q$ definition in \refE{marcum} also holds for complex arguments by analytic continuation \cite{DiBlasio1991}.
\end{itemize}

Thm. \ref{mainThm} allows expressing any result involving $\Phi_3$ in terms of the generalized Marcum-$Q$ function. Besides having archival value, this relationship
greatly facilitates both the evaluation of such results, and any subsequent analysis thereof.

\section{Applications} \label{sec:apps}

The generalized Marcum-$Q$ and the $\Phi_3$ functions appear in a number of communication theory problems. The new connection between these functions presented in Thm. \ref{mainThm} can be therefore applied directly to such problems. For instance, $\Phi_3$ appears in two important distributions: the bivariate Nakagami-$m$ distribution, and the distribution of the minimum eigenvalue of non-central Wishart matrices.
We next exemplify the applicability of Thm. \ref{mainThm} to these specific problems.

\subsection{Bivariate Nakagami-$m$ Distribution} \label{sec:apps:bivariate}

Thm. \ref{mainThm}, together with recent results for the bivariate Nakagami-$m$ distribution in \cite{Lopez2013}, confirm the conjecture made in \cite{Tan1997} and in \cite[p. 174]{Simon05}: the joint CDF of two correlated Nakagami-$m$ variables can be expressed in terms of the generalized Marcum-$Q$ function, thereby completing the landscape of (closed-form) bivariate characterizations of the most common fading distributions.

\begin{corollary} \label{corollary:bivariateCDF}
Let $R_1$ and $R_2$ be two correlated Nakagami-$m$ random variables with positive integer fading index $m$, respective variances $\Omega_1 = \Ex[R_1^2]$ and $\Omega_2 = \Ex[R_2^2]$, and correlation coefficient $\rho$, and let $\hat R_1$ and $\hat R_2$ be the normalized (unit-variance) versions of $R_1$ and $R_2$. Then, the joint CDF of $\hat R_1$ and $\hat R_2$ is given by  
\begin{align} \label{eq:bivariateCDF}
F_{\hat R_1 ,\hat R_2 } \left( {r_1 ,r_2 } \right) &= \frac{{\gamma \left( {m,m \, r_2^2 } \right)}}{{(m - 1)!}} - e^{ - mr_1^2 } \sum\limits_{k = 0}^{m - 1} {\frac{{(m \, r_1^2 )^k }}{{k!}}Q_{1 - k} \left( {r_2 a,r_1 b} \right)} \nonumber \\
& \hspace{-1cm} + e^{ - m \, r_2^2 } \sum\limits_{k = 0}^{m - 1} {\sum\limits_{i = 1}^{m - k} {\sum\limits_{r = 0}^{2(i - 1)} \frac{\left( {1 - \rho } \right)^r }{k!\rho ^r } \left( {\frac{{m \, r_1^2 }}{\rho }} \right)^{k + i - r - 1} \hspace{-2mm} \Acal _r \left( {i,k + i;\frac{{\left( {r_1 r_2 ab} \right)^2 }}{4}} \right) Q_{2 - k - i + r} \left( {r_2 b,r_1 a} \right)} }
\end{align}
where $a = \sqrt{\frac{2m}{1 - \rho }}$, $b = a \sqrt \rho$, $\gamma (\cdot,\cdot)$ is the lower incomplete gamma function \cite{Gradsh00}, and $\Acal_r (\cdot,\cdot;\cdot)$ are the polynomials defined in \refE{phi3_coeffs}. It follows that 
\begin{align} \label{eq:bivariateCDF2}
F_{R_1 ,R_2 } \left( {r_1 ,r_2 } \right) = F_{\hat R_1 ,\hat R_2 } \left( {\frac{{r_1 }}{{\sqrt {\Omega _1 } }},\frac{{r_2 }}{{\sqrt {\Omega _2 } }}} \right) .
\end{align}
\end{corollary}

\begin{proof}
The bivariate Nakagami-$m$ CDF is given in \cite[Eq. 14]{Lopez2013} in terms of the $\Phi_3$ function. Then, \refE{bivariateCDF} is obtained by virtue of Thm. \ref{mainThm}  after normalization of the random variables and some algebraic manipulations.
\end{proof}

\begin{corollary} \label{corollary:bivariateCDF_rayl}
The well-known expression for the bivariate Rayleigh CDF can be recovered from Corollary \ref{corollary:bivariateCDF} by setting $m=1$ and using \cite[Eq. 8.352.1]{Gradsh00} to expand the incomplete Gamma function, yielding
\begin{align} \label{eq:bivariateCDF_rayl}
F_{\hat R_1 ,\hat R_2 } \left( {r_1 ,r_2 } \right) &= 
1 - e^{- r_2^2} - e^{-r_1^2} Q_1 \left(r_2 a, r_1 b \right) + e^{-r_2^2} Q_1 \left(r_2 b, r_1 a \right) 
\end{align}
consistently with \cite[Appendix A]{Schwartz1966}.


\end{corollary}

The closed-form characterization of the bivariate Nakagami-$m$ distribution had remained an open problem for decades. Existing expressions involved infinite summations \cite{Tan1997,Reig2002,Souza2008} or an integral of the product of Marcum-$Q$ functions \cite{Beaulieu2011}. The expression that recently appeared in \cite{Lopez2013}, in terms of $\Phi_3$, has been rewritten by virtue of Thm. \ref{mainThm} into the convenient form given in Corollary \ref{corollary:bivariateCDF}.
The problems whose analysis can benefit from this form include \cite{Simon05,Tan1997,Reig2002,Souza2008,Beaulieu2011,Lopez2013,Abu-Dayya1994,Jakes1974,Zorzi1998,Wang1996}:
\begin{itemize}
\item Determining the impact of fading correlation in dual-diversity reception or transmission \cite{Lopez2013,Abu-Dayya1994}\cite[sect. 5.2.5]{Jakes1974}.
\item Analyzing the level crossing rate and average fade duration of sampled fading envelopes \cite{Lopez2013}.
\item Establishing the transition probabilities for a first-order Markov chain that models a fading process \cite{Zorzi1998,Wang1996}.
This, in turn, can be applied to approximate the envelope of channels with non-independent fading \cite{Wang1996} or to
model the decoding success/failure with automatic repeat-request (ARQ) over successive channel realizations  \cite{Zorzi1998}.
\end{itemize}


\subsection{Minimum Eigenvalue Distribution of Non-Central Wishart Matrices}
\label{sec:apps:wishart}

\begin{definition}[Non-central Wishart matrix]
Let $\X$ be an $n \times m$ ($n \geq m$) random matrix distributed as $\Cn_{n,m} (\mUpsilon , \I_n \otimes \mSigma)$, where $\mSigma$ is the covariance of the independent complex Gaussian row vectors of $\X$, and $\mUpsilon \in \complex^{n \times m}$. Then $\W = \X^H \X$ is a complex non-central Wishart matrix that follows the distribution $\Cw _m (n,\mSigma,\mTheta)$ with $\mTheta = \mSigma^{-1} \mUpsilon^H \mUpsilon$ the non-centrality parameter.
\end{definition}

Concerning the extreme eigenvalues of $\W$, distributional results are available for uncorrelated
central ($\mSigma = \I_m$, $\mUpsilon = \0$), correlated central ($\mUpsilon = \0$), and uncorrelated non-central ($\mSigma = \I_m$) complex Wishart matrices (see, e.g., \cite{Chen1999,Chiani2003,Forrester2007,Maaref2007,Zanella2009}).
However, tractable results for the correlated non-central case had been unavailable until \cite{Dharmawansa2011}, where the minimum eigenvalue distribution of correlated non-central Wishart matrices has been expressed in terms of $\Phi_3$ for some special cases including a square $\X$, i.e., for $m=n$. Specifically, the CDF of $\lambda _{\min}$, the minimum eigenvalue of $\W \sim \Cw _m (m,\mSigma,\mSigma^{-1} \mUpsilon^H \mUpsilon)$, is given when $\mUpsilon$ has rank one
as \cite{Dharmawansa2011}
\begin{align} \label{eq:cdf_min_lambda}
F_{\lambda _{\min } } \left( \lambda  \right) = 1 - \exp \left( { - \eta  - \lambda \tr ( { \mSigma }^{ - 1} )} \right)\Phi _3 \left( {m,m,\eta ,\lambda \mu } \right) 
\end{align}
where $\eta = \tr(\mTheta)$ and $\mu = \tr(\mTheta \mSigma^{-1})$.

Analogous expressions to \refE{cdf_min_lambda}, also in terms of $\Phi_3$, are found for other special cases such as $2 \times 2$ Wishart matrices with arbitrary degrees of freedom, i.e., $m=2$ with arbitrary $n$, or $3 \times 3$ Wishart matrices with $n=4$ degrees of freedom \cite{Dharmawansa2011}.

Using \refE{cdf_min_lambda} and Thm. \ref{mainThm}, we can express the minimum eigenvalue distribution of correlated non-central Wishart matrices as follows.

\begin{corollary} \label{messi}
The CDF of the minimum eigenvalue of $\W \sim \Cw _m (m,\mSigma,\mSigma^{-1} \mUpsilon^H \mUpsilon)$ with rank-one $\mUpsilon$ can be expressed as
\begin{align} \label{eq:cdf_min_lambda_2}
F_{\lambda _{\min } } \left( \lambda  \right) &= 1 - \exp \left( { - \eta  - \lambda \tr({\mSigma}^{ - 1} )} \right) \Gamma (m) \left( {\frac{{\lambda \mu }}{\eta }} \right)^{m - 1} \nonumber \\  
& \quad \times
\sum\limits_{i = 0}^{2(m - 1)} {\Acal _i (m,m;\lambda \mu )} \frac{{\eta ^{i + 1 - m} }}{{(\lambda \mu )^i }} \exp{ \left( \eta  + \frac{\lambda \mu }{\eta } \right)} Q_{2 - m + i} \left( {\sqrt {2\eta } ,\sqrt {2\frac{{\lambda \mu }}{\eta }} } \right)
\end{align}
where the polynomials $\Acal_r (\cdot,\cdot;\cdot)$ are as in \refE{phi3_coeffs}. 
\end{corollary}

Corollary \ref{messi} is restricted to the case of a rank-one non-centrality parameter, which however is typically assumed in multiple-input multiple-output (MIMO) communication systems with a direct line-of-sight path between the transmitter and the receiver \cite{Jayaweera2005,Zhu2009}. Given the complexity of the underlying joint eigenvalue distribution, the CDF in \refE{cdf_min_lambda_2} is remarkably simple, involving only a finite sum of generalized Marcum-$Q$ and elementary functions.

The minimum eigenvalue distribution is important in the analysis of MIMO channels \cite{Burel2002,Heath2005,Narasimhan2003} where the received signal vector is modeled as
\begin{equation}
\y = \H \x + \n 
\end{equation}
with $\H \in \complex^{\Nr \times \Nt}$ the channel matrix containing the gains between the $\Nt$ transmit and $\Nr$ receive antennas, $\n \in \complex^{\Nr \times 1}$ the noise vector, and $\x \in \Scal$ the transmitted signal vector with entries drawn from an alphabet $\Scal$.
The minimum eigenvalue of $\H^H \H$ determines the minimum distance, $d_{\min}$, between the noiseless received signal vectors and, thereby, the error probability of a MIMO maximum likelihood (ML) receiver.
It can be shown that \cite{Burel2002}
\begin{align} \label{eq:dmin}
d_{\min} \geq \sqrt{\lambda_{\min}} \, d_0
\end{align}
where $\lambda_{\min}$ is the minimum eigenvalue of $\H^H \H$ and $d_0$ is the minimum distance between the elements of $\Scal$.
Altogether, the performance of the MIMO ML receiver is strongly linked to the distribution of $\lambda_{\min}$, which is given 
in \refE{cdf_min_lambda_2} for $\H$ having non-central correlated Gaussian entries; this encompasses both Rayleigh and Rice fading with spatial correlation.
Since ML becomes computationally unwieldy 
as the number of antennas or the transmission alphabet cardinality grows, linear and successive cancellation receivers become attractive.
The performance of such receivers
also depends on $\lambda_{\min}$ \cite{Gore2002,Narasimhan2003}.
In fact, the post-receiver SINR of the zero-forcing (ZF), minimum mean square error (MMSE), and Vertical Bell Labs Layered
Space-Time (V-BLAST) receivers satisfies \cite{Narasimhan2003}
\begin{align} \label{eq:SNRmin}
\sinr  \ge \frac{{\Es }}{{\Nt \sigma ^2 }}\lambda _{\min } 
\end{align}
where $\Es$ is the energy per symbol, i.e., $\Ex[\x \x^H] = (\Es / \Nt) \I_{\Nt}$, and $\sigma ^2$ is the noise variance.

Combining Corollary \ref{messi} with (\ref{eq:SNRmin}), the outage probability of MIMO receivers can be analyzed in fairly broad generality.
In addition, Corollary \ref{messi} has further applications, e.g., in the design and analysis of adaptive MIMO multiplexing-diversity switching \cite{Heath2005} or,
 in the context of econometrics, in characterizing the weak instrument asymptotic distribution of the Cragg-Donald statistic \cite{Stock2002}.

%
%
%
%
%


\section*{Acknowledgements}
\label{Ack}

The work of D. Morales-Jimenez and A. Lozano has been supported by the Spanish Government under projects TEC2012-34642, CSD2008-00010 (Consolider-Ingenio) and by the Catalan Government (SGR2009\#70). The work of F. J. Lopez-Martinez is supported by the University of Malaga and by the European Union under Marie-Curie COFUND U-mobility program (ref. 246550). The work of E. Martos-Naya and J. F. Paris has been supported by the Spanish Government-FEDER under projects TEC2010-18451 and TEC2011-25473.   


\appendices

\section{Proof of Lemma \ref{lemma1}} \label{apx:apx1}

The generalized Marcum-$Q$ function can be obtained as the contour integral in the complex plane \cite{Proakis2000}
\begin{align} \label{eq:marcumIntegral}
Q_m \left( {a,b} \right) = \exp \left( { - \frac{{a^2  + b^2 }}{2}} \right)\underbrace { \oint_{\Gamma _0 } {\frac{1}{{p^m (1 - p)}}\exp \left( {\frac{1}{2}\left( {\frac{{a^2 }}{p} + b^2 p} \right)} \right)dp} }_{\Ical_m(a,b)}
\end{align}
where ${\oint_{\Gamma_0}\triangleq \frac{1}{2\pi j}\int_{\Gamma_0}}$ and $\Gamma_0$ is any closed contour enclosing the singularity at $p = 0$ (in a counter-clockwise direction) and no other singularities of the integrand (cf. Fig. \ref{F1}). For convenience, we express the integral $\Ical_m \left(a,b \right)$ in \refE{marcumIntegral} as
\begin{equation} \label{eq:Im}
\Ical_m\left(a,b \right)= \oint_{\Gamma_0} F(p) \exp \left( {\frac{{b ^2 p}}{2}} \right)dp
\end{equation}
with
\begin{equation}
F(p) = \frac{1}{{1 - p}}p^{ - m} \exp \left( {\frac{a^2 }{2p}} \right).
\end{equation}

\begin{figure}[t]
\begin{center}
\includegraphics[width=.6\columnwidth]{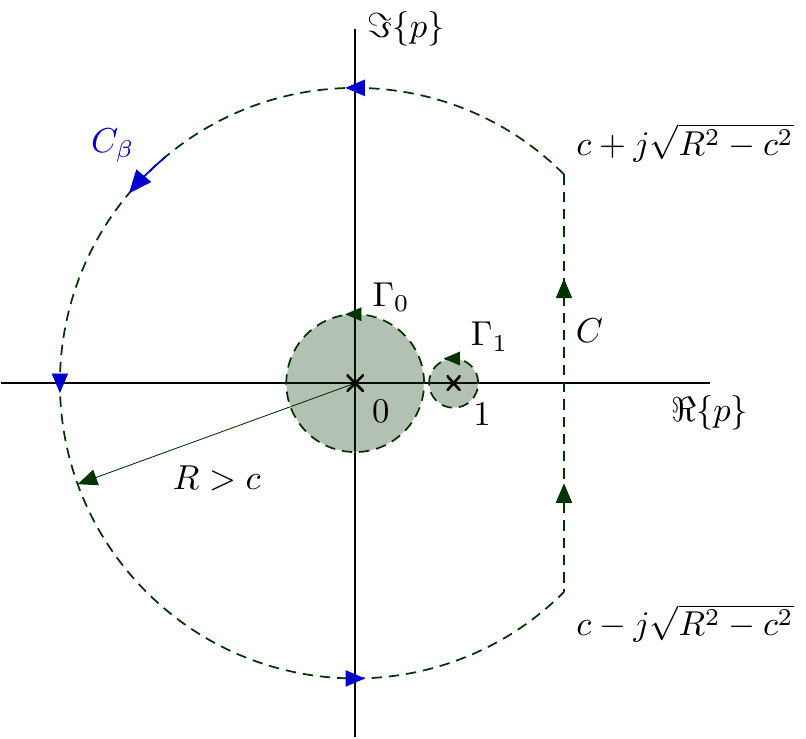}
\caption{Contour integration for integral $\mathcal{I}_{m}$.}
\label{F1}
\end{center}
\end{figure}

Let us consider the contour $C$ depicted in Fig. \ref{F1}, where $c$ is chosen to the right of all the singularities of $F(p)$. Letting $R \to \infty$, the contour integral along $C$ equals
\begin{align}
\label{eqappx}
\oint_C {F (p)\exp \left( {\frac{{b ^2 p}}{2}} \right)dp}  =& \tfrac{1}{{2\pi j}}\int_{c - j\infty }^{c + j\infty } {F (p)\exp \left( {\frac{{b ^2 p}}{2}} \right)dp} + \oint_{C_\beta  } {F (p)\exp \left( {\frac{{b ^2 p}}{2}} \right)dp}.
\end{align}
Alternatively, we can apply the Cauchy-Goursat theorem to obtain
\begin{align}
\label{eqapp2}
\oint_C {F (p)\exp \left( {\frac{{b ^2 p}}{2}} \right)dp} & = \oint_{\Gamma_0 } {F (p)\exp \left( {\frac{{b ^2 p}}{2}} \right)dp} + \oint_{\Gamma_1 } {F (p)\exp \left( {\frac{{b ^2 p}}{2}} \right)dp} 
\end{align}
where $\Gamma_0$ and $\Gamma_1$ are closed contours enclosing the singularities at $p=0$ and $p=1$, respectively. Combining (\ref{eqappx}) and (\ref{eqapp2}),
\begin{align} \label{eq:eqapp3}
&\tfrac{1}{{2\pi j}}\int_{c - j\infty }^{c + j\infty } {F (p)\exp \left( {\frac{{b ^2 p}}{2}} \right)dp}  + \oint_{C_\beta  } {F (p)\exp \left( {\frac{{b ^2 p}}{2}} \right)dp}  = \oint_{\Gamma_0 + \Gamma_1 } {F (p)\exp \left( {\frac{{b ^2 p}}{2}} \right)dp}.
\end{align}
The first integral in \refE{eqapp3} is related to the inverse Laplace transform of $F(p)$ as
\begin{align} \label{eq:Bromwich}
\tfrac{1}{2 \pi j} \int_{c - j\infty }^{c + j\infty } {F(p) \exp \left( {\frac{b ^2 p}{2}} \right) dp} = \left. \Lcal ^{-1} \left\{F(p);t \right\} \right|_{t=\frac{b^2}{2}}.
\end{align}
The integral along $C_\beta$ can be shown to be zero as follows. The modulus of $F(p)$ in $C_\beta$ is 
\begin{equation}
\label{david01}
\left| {F (p)} \right|_{p = R e^{j \theta } } = \frac{1}{{\left| {1 - p} \right|}}\frac{1}{{\left| p \right|^m }}\left| {e^{\delta /p} } \right|
\end{equation}
with $\delta=a^2/2$. Then, we can use the inequalities
\begin{align}
\frac{1}{\left| {1 - p} \right|} & \le \frac{1}{\left| {1 - \left| p \right|} \right|} \mathop  \le \limits_{R > 2} \frac{2}{R} \\
\left| {e^{\delta/p} } \right| & = \left| {e^{\mathrm{Re} \left( {\delta/p} \right)}} \right| \le e^{\left| {\delta/R} \right|} \mathop  \le \limits_{R > R_0 } e^{\left| {\delta/R_0 } \right|}
\end{align}
with arbitrary (finite) $R_0$, to write
\begin{equation}
\left| {F (p)} \right|\mathop  \le \limits_{R > R_0 } \underbrace {2e^{\left| {\delta/R_0 } \right|} }_K \underbrace {R^{ - \left( {m + 1} \right)} }_{R^{ - \ell} }.
\end{equation}
Thus, $\left| {F (p)} \right|_{p=R e^{j\theta}} \le KR^{ - \ell}$ for some $\ell>0$ on $C_{\beta}$ and, according to Jordan's lemma, the integral $\oint_{C_\beta}$ equals 0 as $R \to \infty$. Plugging \refE{Im} and \refE{Bromwich} into \refE{eqapp3} and applying the residue theorem, we arrive at 
\begin{align} \label{eq:rel1}
\Ical_m \left(a,b \right) &= \Lcal ^{ - 1} \left\{ F (p);t \right\}_{t = \frac{b^2 }{2}}  - {\text{Res}}\left\{ {F (p)\exp \left( \frac{b^2 p}{2} \right)} \right\}_{p = 1}
\end{align}
where $\text{Res}\left\{ \cdot \right\}_{p=\dummy}$ denotes the residue at $p=\dummy$. Then, we calculate the residue and use \refE{LaplaceRegPhi3} to solve the inverse Laplace transform, which yields
\begin{align} \label{eq:app1final}
\Ical_m \left(a,b \right)  =& - \left( \frac{b^2 }{2} \right)^m \tilde \Phi _3 \left( {1,m + 1;\frac{b^2}{2},\frac{a^2 b^2}{4} } \right) + \exp \left( {\frac{a^2 + b^2 }{2} } \right).
\end{align}
Further substituting \refE{app1final} in \refE{marcumIntegral} gives the alternative form for the generalized Marcum-$Q$ function,
\begin{align} \label{eq:marcum_phi3_2}
Q_m \left( {a,b} \right) &= 1 - \left( {\frac{{b^2 }}{2}} \right)^m \exp{\left( - \frac{{a^2  + b^2 }}{2} \right)} \tilde \Phi _3 \left( {1,m + 1;\frac{{b^2 }}{2},\frac{{a^2 b^2 }}{4}} \right)
\end{align}
which is valid for any $m \in \real$.

Finally, \refE{marcum_phi3} is obtained by combining (\ref{eq:marcum_phi3_2}) and \refE{marcumNegOrders}, completing the proof.

\section{Proof of Lemma \ref{lemma3}} \label{apx:apx2}

The recursive relationship in \refE{phi3_recursive} is derived by using \refE{LaplaceRegPhi3} and the frequency differentiation property of the Laplace transform,
\begin{align} \label{eq:diffLaplace}
\Lcal \left\{ {t \cdot f\left( t \right)} \right\} =  - \frac{d F\left( s \right)}{ds} . 
\end{align}
Taking the first derivative of \refE{LaplaceRegPhi3}, in light of \refE{diffLaplace},
\begin{align} \label{eq:recursive1}
& t^c \tilde \Phi _3 \left( {b,c;xt,yt} \right) = \nonumber \\
& \quad \quad \Lcal^{ - 1} \left\{ {e^{y/s} \left( {b  x  s^{ - (c + 2)} \left( {1 - \frac{x}{s}} \right)^{ - (b + 1)}  + c  s^{ - (c + 1)} \left( {1 - \frac{x}{s}} \right)^{ - b}  + y  s^{ - (c + 2)} \left( {1 - \frac{x}{s}} \right)^{ - b} } \right)} \right\} 
\end{align}
where $\Lcal^{ - 1} \left\{ \cdot \right\}$ stands for the inverse Laplace transform. Then, with the help of \refE{LaplaceRegPhi3}, we can identify the right-hand side of \refE{recursive1} as a sum of $\tilde \Phi_3$ functions, which allows us to write (after some algebra)
\begin{align} \label{eq:recursive2}
& \tilde \Phi _3 \left( {b,c;xt,yt} \right) = \nonumber \\
& \quad \frac{1}{{(b - 1)  xt}}\left( {\tilde \Phi _3 \left( {b - 1,c - 2;xt,yt} \right) - (c - 2) \tilde \Phi _3 \left( {b - 1,c - 1;xt,yt} \right) - yt \tilde \Phi _3 \left( {b - 1,c;xt,yt} \right)} \right) 
\end{align}
which shows that $\tilde \Phi_3$ can be recursively expressed via lower values of its first argument. Thus, \refE{recursive2} can be recursively applied to yield
\begin{align} \label{eq:recursive3}
\tilde \Phi _3 \left( {b,c;xt,yt} \right) = \sum\limits_{i = 0}^{2(b - 1)} \alpha _i \tilde \Phi _3 \left( {1,c - i;xt,yt} \right) 
\end{align}
where $\alpha _i$ are certain coefficients associated to $\tilde \Phi _3 \left( {1,c - i;xt,yt} \right)$ that can be obtained by working out the recursion \refE{recursive2}. An explicit formula for $\alpha _i$ can be inferred from the first coefficients $\alpha _i$, $i=0,1,2,...$, leading to
\begin{align} \label{eq:coeffs1}
\alpha _i  = \frac{{( - 1)^{b - 1} }}{{\Gamma \left( b \right)}}\frac{{\left( {yt} \right)^{b - 1 - i} }}{{\left( {xt} \right)^{b - 1} }}\sum\limits_{k = 0}^{\left\lfloor {i/2} \right\rfloor} {\frac{{( - 1)^k \left( {b - i + k} \right)_{i - k} \left( {c - i - 1 + k} \right)_{i - 2k} }}{{\left( {i - 2k} \right)!k!}}\left( {yt} \right)^k } .
\end{align}

Finally, \refE{phi3_recursive}-\refE{phi3_coeffs} follow from \refE{recursive3}-\refE{coeffs1} after setting $xt=w$, $yt=z$, and defining the polynomial
$\Acal_i (b,c;z) = (w^{b-1} / z^{b-1-i}) \alpha_i$.


\bibliographystyle{IEEEtran}
\bibliography{IEEEabrv,bib/all}

\end{document}